\documentclass[twocolumn,10pt]{asme2e}

\usepackage{amsmath,graphicx,amsfonts,amssymb,epsfig,mathrsfs,mathtools,subfig}
\usepackage{color}
\usepackage{multirow}
\usepackage{rotating}
\usepackage{graphicx}%
\usepackage{algorithm}
\usepackage{algpseudocode}
\usepackage{tikz}

\newtheorem{assumption}{Assumption}{}
\newtheorem{theorem}{Theorem}{}
{}
{}
{}
\newtheorem{proposition}{Proposition}{}
\newtheorem{remark}{Remark}{}
{}
{}

\algdef{SE}[DOWHILE]{Do}{doWhile}{\algorithmicdo}[1]{\algorithmicwhile\ #1}%

\confshortname{DSCC 2019}
\conffullname{the ASME 2019 Dynamic Systems and Control Conference}

\confdate{October 9-11}
\confyear{2019}
\confcity{Park City}
\confcountry{UT, USA}

\papernum{DSCC2019-9047}

\title{A Stabilizing Control Algorithm for Asynchronous Parallel Quadratic Programming via Dual Decomposition}

\author{Kooktae Lee\thanks{Address all correspondence to this author.} 
    \affiliation{
	Department of Mechanical Engineering, New Mexico Institute of Mining and Technology\\
	Socorro, New Mexico 87801,  Email: kooktae.lee@nmt.edu
    }	
}

\begin{document}

\maketitle    

\begin{abstract}
{\it 
This paper proposes a control algorithm for stable implementation of asynchronous parallel quadratic programming (PQP) through dual decomposition technique. In general, distributed and parallel optimization requires synchronization of data at each iteration step due to the interdependency of data. The synchronization latency may incur a large amount of waiting time caused by an idle process during computation. We aim to mitigate this synchronization penalty in PQP problems by implementing asynchronous updates of the dual variable. The price to pay for adopting asynchronous computing algorithms is the unpredictability of the solution, resulting in a tradeoff between speedup and accuracy. In the worst case, the state of interest may become unstable owing to the stochastic behavior of asynchrony. We investigate the stability condition of asynchronous PQP problems by employing the switched system framework. A formal algorithm is provided to ensure the asymptotic stability of dual variables. Further, it is shown that the implementation of the proposed algorithm guarantees the uniqueness of optimal solutions, irrespective of asynchronous behavior. To verify the validity of the proposed methods, simulation results are presented.
}
\end{abstract}

\begin{nomenclature}
\entry{$\mathbb{R}$}{A set of real numbers}
\entry{$\mathbb{N}$}{A set of natural numbers}
\entry{$\mathbb{N}_0$}{A set of non-negative integers}
\entry{$k$}{A discrete-time index such that $k\in\mathbb{N}_0$}
\entry{$\mathbb{K}_q$}{A set defined by $\{k,k-1,\ldots,k-q+1\}$, where $q\in\mathbb{N}$}
\entry{$\lVert\cdot \rVert_p$}{A p-norm of a given vector or matrix}
\entry{$\rho(\cdot)$}{A spectral radius of a given matrix}
\entry{$^{T}$}{A transpose operator}
\entry{$\otimes$}{A Kronecker product}
\end{nomenclature}

\section*{INTRODUCTION}
Recent advancement of distributed and parallel computing technologies has brought massive processing capabilities in solving large-scale optimization problems. Distributed and parallel computing may reduce computation time to find an optimal solution by leveraging the parallel processing in computation. Particularly, distributed optimization will likely be considered as a key element for large-scale statistics and machine learning problems that have a massive volume of data, currently represented by the word ``big data."
One of the reasons for the preference of distributed optimization in big data is that the size of the data set is so huge that each data set is desirably stored in a distributed manner. Thus, the global objective is achieved in conjunction with local objective functions assigned to each distributed node, which requires communication between distributed nodes in order to attain an optimal solution. 

For several decades, there have been remarkable studies that have enabled us to find an optimal solution in a decentralized fashion, for example, dual decomposition \cite{dantzig1960decomposition,benders1962partitioning}, 
augmented Lagrangian methods for constrained optimization
\cite{fortin2000augmented,aguiar2011augmented},
alternating direction method of multipliers (ADMM) \cite{shi2014linear,chen2016direct}, 
Bregman iterative algorithms for $\ell_1$ problems\cite{bregman1967relaxation,yin2008bregman}, 
Douglas-Rachford splitting \cite{combettes2007douglas}, 
and proximal methods \cite{tseng1997alternating}. 
More details about the history of developments in the methods listed above can be found in the literature \cite{boyd2011distributed}.
In this study, we mainly focus on the analysis of \textit{asynchronous} distributed optimization problems. In particular, we aim to investigate the behavior of asynchrony in the Lagrangian dual decomposition method for parallel quadratic programming (QP) problems, where QP problems refer to the optimization problems with a quadratic objective function associated with linear constraints. This type of QP problem has broad applications including least square with linear constraints, regression analysis and statistics, SVMs, LASSO, portfolio optimization problems, etc. With the implementation of Lagrangian dual decomposition, the original QP problems that are separable can be solved in a distributed sense. For this dual decomposition technique, we will study how asynchronous computing algorithms affect the stability of a dual variable.

Typically, distributed optimization requires synchronization of the data set at each iteration step due to the interdependency of data. For massive parallelism, this synchronization may result in a large amount of waiting time as load imbalance between distributed computing resources would take place at each iteration step. In this case, some nodes that have completed their tasks should wait for others to finish assigned jobs, which causes an idle process of computing resources, and hence a waste of computation time.
This paper attacks this restriction on the synchronization penalty necessarily required in distributed and parallel computing, through the implementation of \textit{asynchronous computing algorithms}. The asynchronous computing algorithm that does not suffer from synchronization latency thus has the potential to break through the paradigm of distributed and parallel optimization. 
Unfortunately, it is not completely revealed yet what is the effect of asynchrony on the stability in the distributed and parallel optimization. Due to the stochastic behavior of asynchrony, an optimal solution for an asynchronous PQP problem may diverge even if it is guaranteed that the synchronous scheme provides stability of the optimal solution. Although Bertsekas \cite{bertsekas1989parallel} introduced a sufficient condition for the stability of general asynchronous fixed-point iterations (see chapter 6.2), which is equivalent to a diagonal dominance condition for QP problems, this condition is known to be very strong and thus conservative, according to the literature \cite{liu2015asynchronous}. 
More recently, a periodic synchronization method has been developed in \cite{lee2016relaxed,lee2016relaxed2,ghosh2018fast}. The key idea is that each node communicates with other nodes at a certain period, instead of synchronizing the data at each iteration step. The dynamics of the given iterative scheme was reformulated by opting for a desirable synchronization period, leading to the synchronization that occurs less frequently than the original parallel computing algorithm. However, the periodic synchronization still requires the idle process time with a given periodicity as an optimal solution is obtained through periodic communications between distributed nodes. Thus this method is not exempt from the synchronization bottleneck issue completely.

In this paper, we develop a stabilizing control algorithm to securely implement the asynchronous computing algorithm for PQP problems. For this purpose, the \textit{switched system} \cite{fang2002stochastic,fang2002stabilization,
lee2018optimal,lee2015performance,lee2015acc,lee2015stability} framework is adopted as a stability analysis tool. In general, the switched system is defined as a dynamical system that consists of a set of subsystem dynamics and a certain switching logic that governs a switching between subsystems. For asynchronous computing algorithms, subsystem dynamics accounts for all possible scenarios of delays in asynchronous computing caused by difference in data processing time or load imbalance in each distributed computing device. Then, a certain switching logic can be implemented to stand for a random switching between subsystem dynamics. Thus, the switched system framework can be used to properly model the dynamics of asynchronous computing algorithms. 
Lee \textit{et al.}\cite{lee2015acc}, for example, introduced the switched system to represent the behavior of asynchrony in massively parallel numerical algorithms. This study applied the switched dynamical system framework in order to analyze the convergence, rate of convergence, and error probability for asynchronous parallel numerical algorithms. 

Based on this switched system framework, we will provide the control method for the stable implementation of asynchronous PQP algorithms with dual decomposition technique. Regardless of asynchronous behavior, the proposed algorithm guarantees the stability of the asynchronous optimal solutions. Moreover, the convergence of asynchronous optimal solutions to synchronous one is also provided, which guarantees the uniqueness of stationary solutions.

\vspace{-0.15in}
\section*{PRELIMINARIES AND PROBLEM FORMULATION}
\subsection*{Duality Problem}
Consider the following QP problem with a linear inequality constraint.
\begin{equation}
\begin{aligned}
&\text{minimize} & f(x)\\
&\text{subject to} & Ax \leq b,\label{eqn:Ax<b}
\end{aligned}
\end{equation}
where $f(x)$ is given by a quadratic form, meaning $f(x) = \dfrac{1}{2}x^{T}Qx + c^{T}x$, the matrix $Q\in\mathbb{R}^{n\times n}$ is a symmetric, positive definite and $c\in\mathbb{R}^n$. Further, $A\in\mathbb{R}^{m\times n}$ and $b\in\mathbb{R}^{m}$.
If we define Lagrangian as $L(x,y) \triangleq f(x) + y^{T}(Ax-b)$, where $y\in\mathbb{R}^m$ is the dual variable or Lagrange multiplier, then the dual problem for above QP problem is formulated as follows.
\newline

\noindent Duality using Lagrangian:
\begin{equation}
\begin{aligned}
&\text{maximize}& \inf_x L(x,y)\label{eqn:inf L(x,y)}\\
&\text{subject to}& y \geq 0.
\end{aligned}
\end{equation}
The primal optimal point $x^{\star}$ is obtained from a dual optimal point $y^{\star}$ as
$
x^{\star} = \underset{x}{\mathrm{argmin}}\: L(x,y^{\star}).
$
By implementing gradient ascent, one can solve the dual problem, provided that $\inf L(x,y)$ is differentiable. In this case, the iteration to find the $x^{\star}$ is constructed as follows:
\begin{equation}
\begin{aligned}
x^{k} &:= \underset{x}{\mathrm{argmin}}\:L(x,y^k),\, y^{k+1} := y^k + \alpha^k(Ax^{k} - b),
\end{aligned}\label{eqn:x^(k+1) = argmin L(x,y^k)}
\end{equation}
where $\alpha^k$ is a step size and the upper script denotes the discrete-time index for iteration.

For the quadratic objective function $f(x)$, the solution for $\underset{x}{\mathrm{argmin}}\:L(x,y^k)$ can be obtained from $\nabla_x L(x,y^k) = 0$ by
\begin{equation}
\begin{aligned}
\underset{x}{\mathrm{argmin}}\:L(x,y^k) &= \nabla_x \left(\dfrac{1}{2}x^{T}Qx + c^{T}x + {y^{k}}^{T}(Ax-b)\right)\\
& = Qx + c + A^{T}y^k = 0,\label{eqn: Qx+c+A^Ty^k=0}
\end{aligned}
\end{equation}
leading to $x^{k} = -Q^{-1}(A^{T}y^k + c)$. 
Plugging this equation into \eqref{eqn:x^(k+1) = argmin L(x,y^k)} yields
\begin{equation}
\begin{aligned}
y^{k+1} &= y^k + \alpha^k\left(A\left(-Q^{-1}(A^{T}y^k+c)\right) -b\right)\\
&= (I-\alpha^kAQ^{-1}A^{T})y^k - \alpha^k(AQ^{-1}c+b).\label{eqn:y^(k+1)=(I-alpha)y^(k+1)-alp(+b)}
\end{aligned}
\end{equation}

\vspace{-0.15in}
\subsection*{Dual Decomposition with Synchronous update}
In this subsection, we consider that $f(x)=\frac{1}{2}x^{T}Qx + c^{T}x$ is \textit{separable}, i.e.,
\begin{align*}
f(x) &= \sum_{i=1}^{N}f_i(x_i)= \sum_{i=1}^{N}\left(\dfrac{1}{2}x_i^{T}Q_ix_i + c_i^{T}x_i\right),
\end{align*}
where $x = [x_1^{T}, x_2^{T}, \hdots, x_N^{T}]^{T}$ and $x_i\in\mathbb{R}^{n_i}$, $i=1,2,\hdots,N$, is a subvector of $x$.
Also, the matrix $A$ in \eqref{eqn:Ax<b} satisfies $Ax = \sum_{i=1}^{N}A_ix_i$, where $A_i$ is such that $A = [A_1, A_2, \hdots, A_N]$.

Then, each $x_i$ is updated by
\begin{equation}
\begin{aligned}
x_i^{k} &:= \underset{x_i}{\mathrm{argmin}}\:L(x_i,y^k) = -Q_i^{-1}(A_i^{T}y^k + c),\\
\end{aligned}\label{eqn:y^(k+1) = y^k with x_i}
\end{equation}
Note that when updating $x_i^{k}$, $i=1,2,\hdots,N$, each value is computed by distributed nodes. Hence, the computation for $x_i^{k}$ can be processed in parallel and then, each value of $x_i^{k}$ is transmitted to the master node to compute $y^{k+1}$ in the gathering stage. Therefore, as in \eqref{eqn:y^(k+1) = y^k with x_i}, updating $y^{k+1}$ requires synchronization of $x_i^{k}$ across all spatial index $i$ at time $k+1$ because $x^{k}$ is obtained by stacking $x_i^{k}$  from $i=1$ to $N$.  
If computing and/or communication delay occurs among one of the index $i$, the process to update $y^{k+1}$ has to be paused until all data is received from distributed nodes. 
This implies that the more parallel computing we have, the more delays may take place, resulting in a large amount of idle process time. Consequently, this idle time for synchronization becomes dominant compared to the pure computation time to solve the QP problem in parallel.
To mitigate or avoid this type of restriction that severely affects on the performance to obtain an optimal solution, we introduce \textit{asynchronous computing algorithm} in the following subsection.

\vspace{-0.15in}
\subsection*{Dual Decomposition with Asynchronous update}
In order to alleviate this synchronization penalty, we consider asynchronous updates of dual variable $y$. In this case, the master node to compute $y^{k+1}$ does not need to wait until all $x_i^{k}$ is gathered. Rather, it proceeds with the specific value $x_i$ saved in the buffer memory. Thus, $y$ value is updated asynchronously. To model the asynchronous dynamics of dual decomposition, we consider the new state vector for the asynchronous model $\tilde{x}^{k} := [{x_1^{k_1^*}}^{T}, {x_2^{k_2^*}}^{T}, \hdots, {x_N^{k_N^*}}^{T}]^{T}$,
where $k_i^* \in \mathbb{K}_q$, $i=1,2,\hdots, N$, denotes the delay term that may take place due to computation and/or communication delays across $N$ distributed nodes, and the variable $q\in\mathbb{N}$ denotes the cardinality of a set $\mathbb{K}_q$.
In Fig. \ref{fig:2}, we illustrated the conceptual schematic of asynchronous update.
\begin{figure}[!h]
\centering
\includegraphics[scale=0.4]{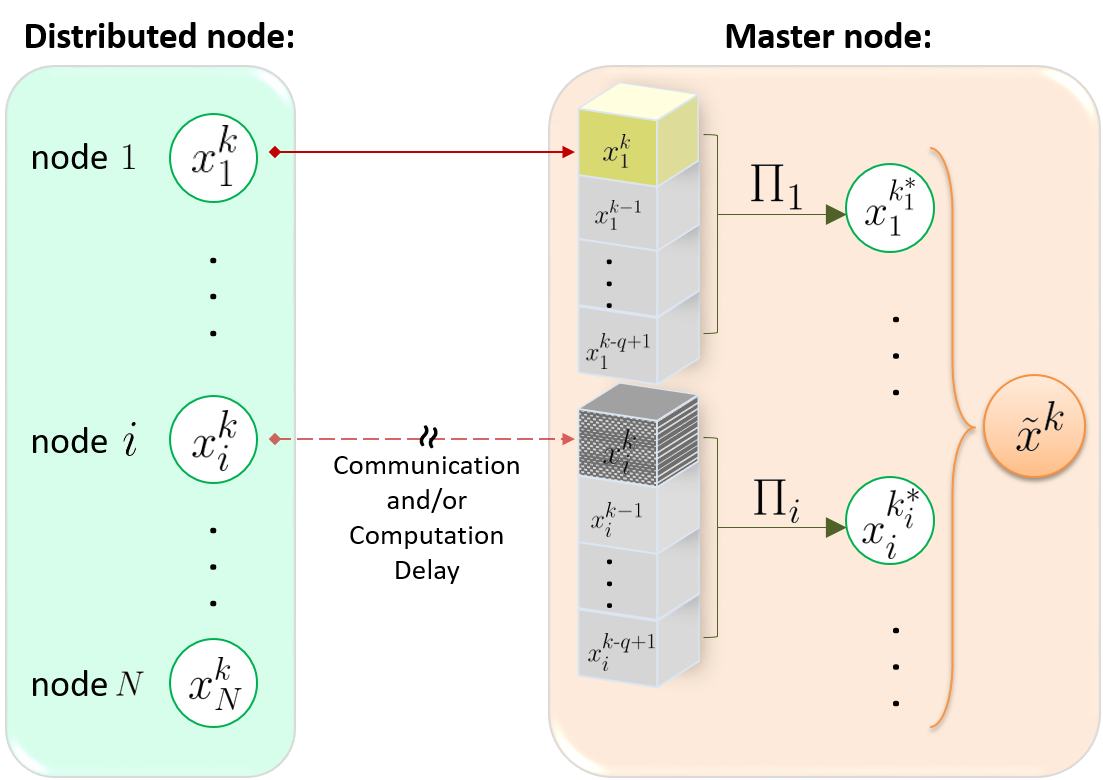}
\begin{align*}
q\:\: &: \text{the cardinality of a set } \mathbb{K}_q:=\{x^k, x^{k-1},\hdots, x^{k-q+1}\}\\
x_i^k &: \text{the value of } x_i \text{ at time }k\\
{x_i^{k_i^*}} &: \text{a random variable such that }{k_i^*} \in \mathbb{K}_q\\
\Pi_i &:= [\Pi_i(1), \Pi_i(2),\hdots, \Pi_i(q)], \text{ where $\Pi_i(j)$ represents}\\ 
& \quad \text{the $j^{th}$ modal probability for }{k_i^*}\\
\tilde{x}^k &:= [{x_1^{k_1^*}}^{T},{x_2^{k_2^*}}^{T}, \hdots, {x_N^{k_N^*}}^{T}]^{T}
\end{align*}
\caption{The schematic of the asynchronous computing algorithm for a PQP problem. In this figure, the delay is bounded by $k-q+1\leq k_i^* \leq k$, $\forall i$. Each node has the probability $\Pi_i$ to represent randomness of delays.}\label{fig:2}
\end{figure}

For this asynchronous case, $y$-update is given by
\begin{equation}
\begin{aligned}
y^{k+1} &= y^k + \alpha^k(A\tilde{x}^{k} - b)= y^k + \sum_{i=1}^{N}\left(\alpha_i^kA_i\tilde{x_i}^{k} - \dfrac{1}{N}\alpha_i^k b\right),\label{eqn:y^(k+1) = y^k + sum_tilde_x_i^(k+1)-b}
\end{aligned}
\end{equation}
where $\alpha_i^k$ is the step size for the index $i$. 
Although $\alpha_i^k$ may vary at each time step, we let $\alpha_i^k$ be a constant value, denoted by $\alpha_i$, for simplicity.

In the real implementation of distributed and parallel optimization, $k_i^*$ varies from distributed nodes and also changes over each iteration. Thus, $x_i^{k_i^*}$ becomes a random vector affected by a random variable $k_i^*\in\mathbb{K}_q$.

Starting from \eqref{eqn:y^(k+1) = y^k + sum_tilde_x_i^(k+1)-b}, with the definition of the set 
$\mathcal{S}^k := \{k_i^*| k_i^* = k\}$ and the symbol $\Phi_i:=-\alpha_iA_iQ_i^{-1}A_i^{T}$, the state dynamics of the asynchronous algorithm can be rewritten by

\begin{align}
&y^{k+1} 
 = \left(I + \sum_{i\in\mathcal{S}^{k}} \Phi_i\right)y^{k} + \left(\sum_{i\in\mathcal{S}^{k-1}} \Phi_i\right)y^{k-1} + \cdots \quad\quad \left(\text{from \eqref{eqn:y^(k+1) = y^k with x_i}}\right)\nonumber\\
& + \left(\sum_{i\in\mathcal{S}^{k-q+1}} \Phi_i \right)y^{k-q+1} + \left(\sum_{i=1}^{N} -\alpha_iA_iQ_i^{-1}c - \dfrac{1}{N} \alpha_ib\right).\label{eqn:2.14}
\end{align}

The above equation is simplified by
\begin{equation}
\begin{aligned}
y^{k+1} &= R_{11}^ky^k + R_{12}^ky^{k-1} + \cdots + R_{1q}^ky^{k-q+1} + B, \label{eqn:async. y^(k+1) update}
\end{aligned}
\end{equation}
where
\begin{align*}
R_{1j}^k &=
\left\{
\begin{array}{ll}
I+\displaystyle \sum_{i\in\mathcal{S}^{k}}\Phi_i, & j=1\\
\\
\displaystyle \sum_{i\in\mathcal{S}^{k-j+1}}\Phi_i, & j\neq1
\end{array}
\right. ,\,\,\, B := \left(\sum_{i=1}^{N} -\alpha_iA_iQ_i^{-1}c - \dfrac{1}{N} \alpha_i b\right). 
\end{align*}
Remind that the evolution of $y^k$ forms a stochastic process due to the time-varying structure of $R_{1j}^k$, caused by the random variable $k_i^{\star}$. 

For stability of general fixed-point iteration, Bertsekas and Tsitsiklis \cite{bertsekas1989parallel} developed  the sufficient condition, which guarantees that the dual variable $y^k$ under the asynchronous computation \eqref{eqn:async. y^(k+1) update} is stable, if
\begin{align}
\rho\left(\Big| I + \sum_{i=1}^{N}\Phi_i\Big| \right) < 1, \label{eqn: Bertsekas stability condition}
\end{align}
where $|\cdot|$ represents the absolute value of all elements for a given matrix.

However, the above stability condition may lead to conservatism, and thus the asynchronous scheme may become stable even if the condition \eqref{eqn: Bertsekas stability condition} is violated. This concern is also addressed in \cite{liu2015asynchronous}.
The primary goal of this paper is therefore to 1) develop the stability condition for the asynchronous PQP algorithm with less conservatism; 2) design a control algorithm that stabilizes the asynchronous scheme. For this purpose, we adopt a \textit{switched linear system (or jump linear system, interchangeably)} framework that will be introduced in the next section.

\section*{MODELING: A SWITCHED SYSTEM APPROACH}
In order to solve the dual decomposition problem that has random delays across distributed nodes, we define a new augmented state $Y^{k} := [{y^k}^{T}, {y^{k-1}}^{T}, \hdots, {y^{k-q+1}}^{T}]^{T}$.
Then, the following recursive dynamics can be used to denote the behavior of asynchronous computing algorithm.
\begin{eqnarray}
&Y^{k+1} = W^{k}Y^{k} + C,\\
&\text{where}\quad W^k =
\begin{bmatrix}
R_{11}^{k} & R_{12}^{k} & R_{13}^{k} &\cdots & R_{1q}^{k}\\
I & 0 & \cdots & & 0\\
0 & I & 0 & \cdots & 0\\
\vdots &&\ddots&&\vdots \\
0 & 0 &&I & 0
\end{bmatrix}
,\quad
C =
\begin{bmatrix}
B\\
0\\
0\\
\vdots\\
0
\end{bmatrix}.\label{eqn:W^k}
\end{eqnarray}

In fact, the structure of the time-varying matrix $W^{k}$ is not arbitrary but has a finite number of forms, given by $q^N$, which counts all possible scenarios to distribute $N$ numbers of $\Phi_i$, $i=1,2,\hdots,N$, matrices defined in \eqref{eqn:2.14} into the number $q$. In the switched system, this number is referred to as the \textit{switching mode number}, and we particularly denote this number with the symbol $\eta$. For instance, when $q=2$ and $N=2$, the switching mode number is given by $\eta=2^2=4$. Thus, at each time $k$, the matrix $W^{k+1}$ has one of the following form:
\begin{equation}
\begin{aligned}
W_1 = \begin{bmatrix}
I+\Phi_1 + \Phi_2 &\quad 0\\
I &\quad 0
\end{bmatrix}, \qquad
W_2 = \begin{bmatrix}
I + \Phi_1 &\quad \Phi_2\\
I &\quad 0
\end{bmatrix}, \\ 
W_3 = \begin{bmatrix}
I + \Phi_2 &\quad \Phi_1\\
I &\quad 0
\end{bmatrix}, \qquad
W_4 = \begin{bmatrix}
I &\quad \Phi_1  +\Phi_2\\
I &\quad 0
\end{bmatrix}.
\end{aligned}\label{eqn: modal matrcies W_i}
\end{equation}
Then, only one among a set of all matrices $\{W_{r}\}_{r=1}^{\eta}$ will be used at each time $k$ to update the system state $Y^k$, which results in the \textit{switched linear system} structure as follows:
\begin{align}
Y^{k+1} = W_{\sigma_{k}}Y^k + C,\:\:\:\sigma_k \in \{1,2,\hdots,\eta\},\:k\in\mathbb{N}_0,\label{eqn:Y^(k+1) = W_sigmaY^(k+1) + Z}
\end{align}
where $\{\sigma_k\}_{k=0}^{\infty}$ denotes the switching sequence that describes how the asynchrony takes place. Then, the switching probability $\Pi^{k} := \Pi_1^{k}\otimes \Pi_2^{k}\otimes \cdots\otimes \Pi_N^{k} = [\pi_1^{k}, \pi_2^{k},$ $ \hdots, \pi_{\eta}^{k}]$, where $\Pi_i^{k}$ represents the probability for ${k_i^*}$ at time $k$ as depicted by Fig. \ref{fig:2}, determines which mode $\sigma_k$ will be utilized at each time step. 
(Note that $\Pi_i^{k}$ and hence $\Pi^{k}$ are not necessarily to be stationary.)
In this case, the switched linear system is named by ``stochastic switched linear system" or ``stochastic jump linear system"  \cite{lee2015performance} because the switching is a stochastic process.
The benefit when applying this stochastic switched linear system structure is that the asynchrony in \eqref{eqn:async. y^(k+1) update} is naturally taken into account by this framework. Then, the randomness of the asynchronous computing algorithm is represented by certain switching logic.

\vspace{-0.15in}
\section*{STABILITY ANALYSIS AND CONTROL ALGORITHM}
In this section, stability of the state $y^k$ for the asynchronous model will be studied under the switched system framework. For several decades, the stability results for the switched system with stochastic jumping parameters have been well established in the literature \cite{lee2015performance,feng1992stochastic,costa1993stability,fang2002stochastic}.
For instance, the stability is determined by testing the condition $\rho\left(\sum_{i=1}^{\eta}\pi_i(W_i\otimes W_i) \right)<1$ for the i.i.d. switching process (see Corollary 2.7 in \cite{fang2002stochastic}) or $\rho\left(\left(P^{T}\otimes I\right)\text{diag}(W_i\otimes W_i) \right)<1$, where $P$ is the Markov transition probability matrix and diag$(W_i\otimes W_i)$ is the block diagonal matrix formed by $W_i\otimes W_i$, $i=1,2,\ldots,\eta$, for the Markov switching process (see Theorem 1 and 2 in \cite{costa1993stability}).
However, these methods cannot be directly applicable to the stability analysis of massively parallel computing algorithm under asynchronism. One of the biggest concerns is the notorious \textit{scalability problem} (also known as the curse of dimensionality problem) that causes computational intractability, due to extremely large numbers of switching modes $\eta$. 
In \cite{lee2015acc}, Lee \textit{et al.} first addressed this issue, which is briefly explained as follows.

\begin{remark}\label{remark:3.1}\textit{\textbf{(Scalability problem)}}
Although the stochastic switched linear system framework is suitable for modeling the dynamics of the asynchronous algorithm in PQP problems, it brings about an extremely large number of switching modes, causing \textit{computational intractability}.
For example, even if $q=2$ and $N=20$, it results in the switching mode number given by $\eta=q^{N} = 2^{20}$.  In the real implementation, it is infeasible to store such large numbers of matrices for the purpose of stability analysis. 
\end{remark}

In the following, we introduce new stability result for the asynchronous PQP algorithm to avoid the scalability problem stated above.

\begin{proposition}\label{prop: stability}
Consider the switched linear system \eqref{eqn:Y^(k+1) = W_sigmaY^(k+1) + Z} whose modal matrices have the structure as in \eqref{eqn:W^k} with a switching sequence $\{\sigma_k\}$ governed by any stochastic processes (e.g., i.i.d. or Markov).
This switched linear system is stable \textit{irrespective of} a switching sequence $\{\sigma_k\}$, if for any $k\in\mathbb{N}_0$,
\begin{align}
\sum_{j=1}^{q} \Big\lVert R_{1j}^{k}\Big\rVert_p < 1, \label{eqn: stability condition}
\end{align}
where $q\in\mathbb{K}_q$, $R_{1j}^k$ is the block matrix as in \eqref{eqn:W^k}.
\end{proposition}
\begin{proof}
For simplicity, consider the case for $q=2$. The most general case for any $q\in\mathbb{N}$ can be proved in the same context. 

When $k=0$, $Y^1$ is computed from \eqref{eqn:Y^(k+1) = W_sigmaY^(k+1) + Z} by
\begin{align*}
Y^1 = \begin{bmatrix}
R_{11}^{0} & R_{12}^{0}\\
I & 0
\end{bmatrix}Y^0 + C,
\end{align*}
where $Y^0 = [{y^0}^{T}, {y^0}^{T}]^{T}$. 
Since $y^1 = \left(R_{11}^0 + R_{12}^{0}\right)y^0$ from the above equation, 
by the sub-multiplicativity and sub-additivity of p-norm we have
\begin{align*}
\lVert y^1 \rVert_p &= \lVert \left(R_{11}^0 + R_{12}^0\right)y^0 \rVert_p
\,\leq \lVert R_{11}^0 + R_{12}^0 \rVert_p \cdot \lVert  y^0 \rVert_p\\
&\leq \left(\lVert R_{11}^0 \rVert_p  + \lVert R_{12}^0 \rVert_p\right) \cdot \lVert  y^0 \rVert_p,
\end{align*}
resulting in the contraction of $y^1$, if $\lVert R_{11}^0\rVert_p + \lVert R_{12}^0\rVert_p < 1$.

When $k=1$, the structure of the matrix $W^1W^0$ is obtained by
\begin{align*}
W^{1}W^{0}&= \begin{bmatrix}
R_{11}^{1} & R_{12}^{1}\\
I & 0
\end{bmatrix}
\begin{bmatrix}
R_{11}^{0} & R_{12}^{0}\\
I & 0
\end{bmatrix}=\begin{bmatrix}
R_{11}^{1}R_{11}^{0} + R_{12}^{1} &\quad R_{11}^{1}R_{12}^{0}\\
R_{11}^{0} & R_{12}^{0}
\end{bmatrix},
\end{align*}
which yields $y^2 = \left(R_{11}^1R_{11}^{0} + R_{12}^1 +  R_{11}^{1}R_{12}^{0} \right)y^0$.
If $\sum_{j=1}^{q} \Big\lVert R_{1j}^{k}\Big\rVert_p < 1$ for $k=0,1$, then $y^2$ is contracting by the following result:
\begin{align*}
&\Big\lVert 
R_{11}^1R_{11}^{0} + R_{12}^1 +  R_{11}^{1}R_{12}^{0}
\Big\rVert_p 
= \Big\lVert 
R_{11}^1\left(R_{11}^{0} + R_{12}^{0}\right) + R_{12}^1
\Big\rVert_p\\
&\leq \lVert R_{11}^{1}\rVert_p\left(\lVert R_{11}^{0}\rVert_p + \lVert R_{12}^{0}\rVert_p\right) + \lVert R_{12}^{1}\rVert_p  < 1.
\end{align*}

When $k=2$, we have
\begin{align*}
&W^{2}W^{1}W^{0}=\\
& \begin{bmatrix}
R_{11}^{2}(R_{11}^{1}R_{11}^{0} + R_{12}^{1}) + R_{12}^{2}R_{11}^{0} &\quad R_{11}^{2}R_{11}^{1}R_{12}^{0} + R_{12}^{2}R_{12}^{0}\\
R_{11}^{1}R_{11}^{0} + R_{12}^{1} & R_{11}^{1}R_{12}^{0}
\end{bmatrix}.
\end{align*}

Similar to the case for $k=0$ and $k=1$, $y^3$ is contracting, if $\sum_{j=1}^{q} \Big\lVert R_{1j}^{k}\Big\rVert_p < 1$ for $k=0,1,2$ as this condition leads to
\begin{align*}
&\Big\lVert 
R_{11}^{2}(R_{11}^{1}R_{11}^{0} + R_{12}^{1}) + R_{12}^{2}R_{11}^{0} + R_{11}^{2}R_{11}^{1}R_{12}^{0} + R_{12}^{2}R_{12}^{0}
\Big\rVert_p \\
&\leq \lVert R_{11}^{2}\rVert_p\bigg(\lVert R_{11}^{1}\rVert_p\Big(\lVert R_{11}^{0}\rVert_p + \lVert R_{12}^{0}\rVert_p\Big) + \lVert R_{12}^{1}\rVert_p\bigg)\\
& \quad + \lVert R_{12}^{2}\rVert_p \left( \lVert R_{11}^{0}\rVert_p + \lVert R_{12}^{0}\rVert_p \right) < 1.
\end{align*}
Thus, by induction it can be concluded that $y^k$ is contracting if the condition \eqref{eqn: stability condition} holds for all $k\in\mathbb{N}_0$.
\end{proof}

Proposition \ref{prop: stability} states that the switched linear system \eqref{eqn:Y^(k+1) = W_sigmaY^(k+1) + Z} is stable if $\sum_{j=1}^{q} \big\lVert R_{1j}^{k}\big\rVert_p < 1$ for \textit{all} $k\in\mathbb{N}_0$. Unfortunately, it is not guaranteed that \eqref{eqn: stability condition} holds during the entire iteration process of asynchronous computations, and hence the asynchronous scheme may become unstable.
To circumvent this concern, we synthesize the following control algorithm.
\begin{algorithm}[h]
\caption{A stabilizing control algorithm}\label{algorithm:1}
\begin{algorithmic}[1]
\State initialize $y^0$, $\epsilon$, $k\gets 0$, $\zeta\gets 1$
\Do
\If {$\sum_{j=1}^{q} \big\lVert R_{1j}^{k}\big\rVert_p < 1$}
\State update $y^{k+1}$ using \eqref{eqn:async. y^(k+1) update}
\State $\zeta\gets 0$
\Else {}
\State $y^{k+1}\gets y^k$
\State $\zeta\gets 1$
\EndIf
\State $k\gets k+1$
\doWhile{$\left\lVert y^{k+1} - y^{k}\right\rVert_p \geq \epsilon$ \textbf{ or } $\zeta\neq 0$}
\end{algorithmic}
\end{algorithm}
In Algorithm \ref{algorithm:1}, $\zeta$ is defined as a masking index to prevent immediate termination of the code when \eqref{eqn: stability condition} is not satisfied. In this way, the code keeps running until both conditions $\left\lVert y^{k+1} - y^{k}\right\rVert_p < \epsilon$ for some positive constant $\epsilon$ and $\zeta=0$ are met. 

The proposed algorithm ensures the secure implementation of the asynchronous PQP algorithm for stability. If the condition \eqref{eqn: stability condition} is satisfied, then $y^k$ is updated asynchronously with the guaranteed convergence of $y^k$ by Proposition \ref{prop: stability}. Otherwise, $y^{k+1}$ remains the same as indicated by the line 7 in Algorithm \ref{algorithm:1}. 
To avoid the issue that Algorithm \ref{algorithm:1} falls into an infinite loop caused by the violation of the condition \eqref{eqn: stability condition} for \textit{every} $k\in\mathbb{N}_0$, we propose the following assumption.
\begin{assumption}
For any \textit{finite} time interval $[k_1, k_2]$, where $k_1,k_2\in\mathbb{N}_0$, there exists time $k\in[k_1,k_2]$ such that \eqref{eqn: stability condition} is satisfied.\label{assump: finite time}
\end{assumption}

Above assumption is not restrictive because the master node will receive more concurrent data for $x_i^k$ from its slave nodes while waiting for the condition \eqref{eqn: stability condition} to be satisfied. 
Thus \eqref{eqn: stability condition} will be eventually satisfied as the iteration count increases.

With the implementation of Algorithm \ref{algorithm:1}, the stability of the asynchronous PQP algorithm through  dual decomposition is then provided as follows.

\begin{theorem}\label{thm: stability}\textit{\textbf{(Asymptotic stability)}}
Consider the asynchronous dual decomposition update \eqref{eqn:async. y^(k+1) update} for PQP problems. With the implementation of Algorithm \ref{algorithm:1}, the quantity of interest $y^k$ is asymptotically stable under Assumption \ref{assump: finite time}.
\end{theorem}
\begin{proof}
The result is a direct consequence of Proposition \ref{prop: stability} under Assumption \ref{assump: finite time}.
\end{proof}

Although Theorem \ref{thm: stability} guarantees the asymptotic stability of \eqref{eqn:async. y^(k+1) update}, there is no evidence that $y^k$ under asynchronous communications converges to the solution obtained by the synchronous scheme. This is one of the critical issues in asynchronous computing communities (e.g., see \cite{lee2019uniqueness}) as asynchrony may induce inaccuracy in solutions. We provide the following theorem to show that asynchronous optimal solutions computed by Algorithm \ref{algorithm:1} are not affected by asynchrony at all and thus are identical to the synchronous one.
\begin{theorem}\textit{\textbf{(Uniqueness of optimal solutions)}}
Asynchronous solutions for \eqref{eqn:async. y^(k+1) update} with the implementation of Algorithm \ref{algorithm:1} are identical to the unique synchronous solution.
\end{theorem}
\begin{proof}
The proof begins with the synchronous case of which dynamics is given by
\begin{align}
Y_{\text{sync.}}^{k+1} = W_{\text{sync.}}Y_{\text{sync.}}^{k} + C,\label{eqn:Y_sync}
\end{align}
\vspace{-0.35in}
\begin{align*}
\text{where} \quad W_{\text{sync.}} &=
\begin{bmatrix}
R_{\text{s}} & 0 & 0 &\cdots & 0\\
I & 0 & 0 & \cdots  & 0\\
0 & I & 0 & \cdots & 0\\
\vdots &&\ddots&&\vdots \\
0 & 0 &&I & 0
\end{bmatrix},\quad
R_{\text{s}}  = I + \sum_{i=1}^{N}\Phi_i.
\end{align*}

Note that we also have $R_{\text{s}} = \sum_{j=1}^{q}R_{1j}^k$ for any $k\in\mathbb{N}_0$.
If we denote $Y^{\star}$ and $Y_{\text{sync.}}^{\star}$ as the optimal solution for asynchronous and synchronous scheme, respectively, the error is computed by \eqref{eqn:Y_async} and \eqref{eqn:Y_sync} as follows:
\begin{equation}
Y^{\star} - Y_{\text{sync.}}^{\star} = W^kY^{\star} - W_{\text{sync.}}Y_{\text{sync.}}^{\star}.\label{eqn: Y^* - Y_sync}
\end{equation}
Since the optimal solution is stationary, we have $Y^{\star} = [{y^{\star}}^{T},{y^{\star}}^{T},\ldots, {y^{\star}}^{T}]^{T}$ and $Y_{\text{sync.}}^{\star} = [{y_{\text{sync.}}^{\star}}^{T},{y_{\text{sync.}}^{\star}}^{T},\ldots, {y_{\text{sync.}}^{\star}}^{T}]^{T}$, which from the first row block of \eqref{eqn: Y^* - Y_sync} leads to
\begin{align*}
y^{\star} - y_{\text{sync.}}^{\star} &= \left(\sum_{j=1}^{q}R_{1j}^k\right)y^{\star} - R_{\text{s}}\,y_{\text{sync.}}^{\star}= R_{\text{s}}(y^{\star}-y_{\text{sync.}}^{\star}).
\end{align*}
As the matrix $R_{\text{s}}-I$ is positive definite by the given structure of $R_s$ and the positive definiteness of all $\Phi_i$ matrices, the above equation satisfies $y^{\star}=y_{\text{sync.}}^{\star}$.
\end{proof}

\section*{SIMULATIONS}
\begin{figure*}[h!]
\centering
\subfloat[$y_1^k$]{\includegraphics[scale=0.25]{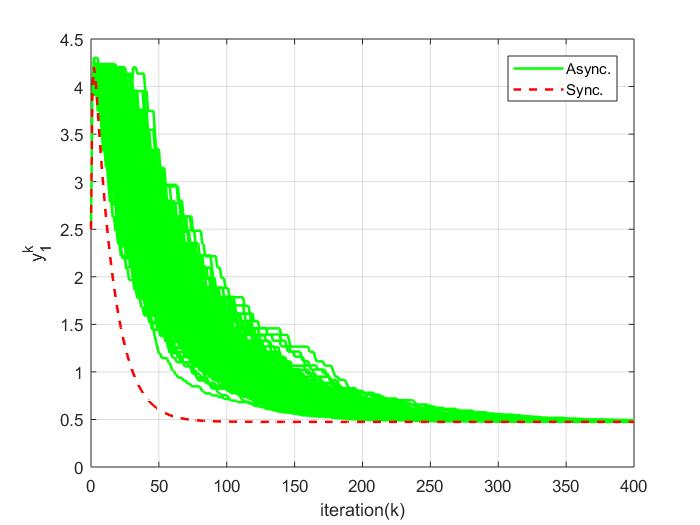}}
\subfloat[$y_2^k$]{\includegraphics[scale=0.25]{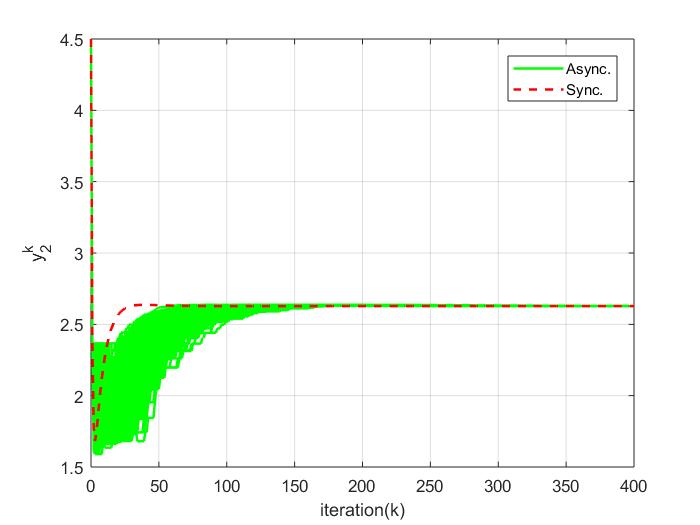}}
\subfloat[$y_3^k$]{\includegraphics[scale=0.25]{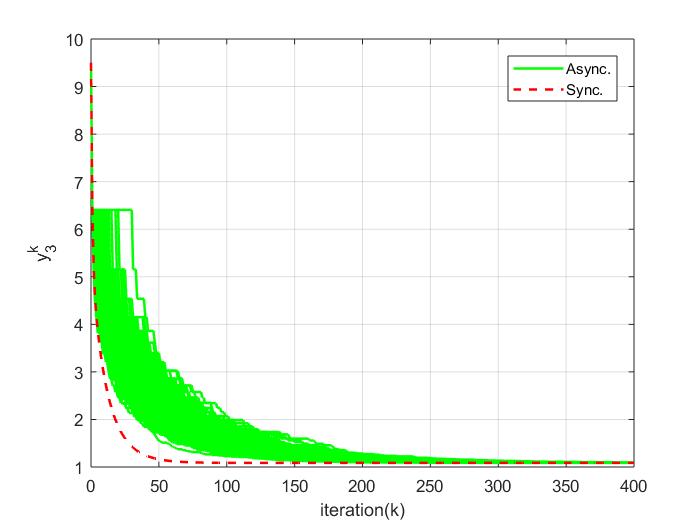}}
\caption{Simulation results for the PQP problem with the proposed asynchronous computing algorithm. The solid lines represent the trajectories of $y_i$ with total 300 Monte Carlo simulations (initial value was deterministically given for all cases). The dashed lines correspond to the synchronous case.}\label{fig:3}
\end{figure*}

To validate the stability of the proposed control algorithm for asynchronous PQP problems as well as the uniqueness of optimal solutions, simulation results are provided in this section. Let us consider the following PQP problem:
\begin{eqnarray*}
&\text{minimize} &\qquad \sum_{i=1}^{N}\left(\dfrac{1}{2}x_i^{T}Q_ix_i + c_i^{T}x_i\right)  \\
&\text{subject to}&\qquad A_ix_i \leq b_i,\qquad\qquad i=1,2,\hdots, N.
\end{eqnarray*}

The positive definite matrix $Q_i$, the matrix $A_i$, and the vectors $b_i$ and $c_i$ were generated by the random number generator in MATLAB. 
The dimension of matrices and vectors are set to be: $Q_i\in\mathbb{R}^{n\times n}, A_i\in\mathbb{R}^{m\times n}$, $c_i\in\mathbb{R}^{n\times 1}$, and $b_i\in\mathbb{R}^{m\times 1}$, $i=1,2,\hdots, N$, where $m=3$, $n=10$, $N=150$. 
The step size is given by $\alpha_i=1.5\times 10^{-6}$, $\forall i$ and $q=30$.
In this case, it turns out that $\rho\left(\lvert I + \sum_{i=1}^{N}\Phi_i\rvert\right)=1.077 > 1$, which does not satisfy the stability condition \eqref{eqn: Bertsekas stability condition}. Therefore, the stability of asynchronous dual decomposition for the given PQP problem is not guaranteed by this condition.

To test the validity of the proposed method, multiple simulations are carried out. At each iteration, the random variable $k_i^*$ is determined by 
the i.i.d. probability $\Pi_i^k = \left[\Pi_i^k(j)\right]=\left[\dfrac{e^{-1.2\times j}}{\sum_{j=1}^{q}e^{-1.2\times j}}\right]$, $j=1,\ldots,q$ for all $i$ and $k$. Notice that the proposed method is, however, applicable to any stochastic processes since the proposed result is independent of a switching sequence $\{\sigma_k\}$.
In Fig. \ref{fig:3}, total $300$ numbers of trajectories for dual variables $y_1$, $y_2$, and $y_3$ are represented by  solid lines. As $y$-update is governed by the given i.i.d. process, the trajectories are different from each other, resulting in the spread of the trajectories in the transient time.  

The computational complexity was the major concern as described in Remark \ref{remark:3.1}, which caused total $\eta=q^N = 30^{150}$ numbers of switching modes in this case. With the developed stabilizing control algorithm, this complexity issue is avoided along with the guaranteed asymptotic stability, irrespective of asynchrony. It is also verified in simulation results that all optimal solutions obtained by the proposed algorithm converged to that by the synchronous scheme shown by the dashed lines in Fig. \ref{fig:3}. Consequently, the proposed stabilizing control algorithm guarantees both the stability as well as the uniqueness of optimal solutions.

\section*{CONCLUSION}
In this paper, we proposed the stabilizing control algorithm for stable implementation of asynchronous PQP algorithm via dual decomposition. To analyze the behavior of asynchrony in distributed and parallel computing, the switched system framework was introduced. 
Although the switching mode number increases exponentially with respect to the distributed node numbers, the developed method, which is not affected by a switching sequence, provides an effective way to analyze the stability without any concerns for such a scalability issue. The uniqueness of the asynchronous optimal solution is also guaranteed for the proposed method. Finally, simulation results are presented to support the validity of the developed methods.


\small
\bibliographystyle{abbrv}
\bibliography{reference}

\end{document}